\documentclass[reprint,amsfonts, amssymb, amsmath,  showkeys, nofootinbib, prx, superscriptaddress, twocolumn,longbibliography]{revtex4-2}
\usepackage{float}
\makeatletter
\let\newfloat\newfloat@ltx
\makeatother
\usepackage[english]{babel}
\usepackage[utf8]{inputenc}
\usepackage{graphics}
\usepackage{selinput}
\usepackage[normalem]{ulem}
\usepackage[shortlabels]{enumitem}

\usepackage{braket}
\usepackage{amsthm}
\usepackage{mathtools}
\usepackage{physics}
\usepackage{xcolor}
\usepackage{graphicx}
\usepackage[left=16mm,right=16mm,top=35mm,columnsep=15pt]{geometry} 
\usepackage{adjustbox}
\usepackage{placeins}
\usepackage[T1]{fontenc}
\usepackage{lipsum}
\usepackage{csquotes}
\usepackage{bm}

\usepackage[linesnumbered,ruled,vlined]{algorithm2e}
\SetKwInput{kwInit}{Init}

\def\HC{\mathcal{H}}

\usepackage{mathrsfs} 

\def\ad{^{\dagger}}

\newcommand{\fsnull}[1]{}
\newcommand{\old}[1]{}

\usepackage[makeroom]{cancel}
\usepackage[toc,page]{appendix}
\usepackage[colorlinks=true,citecolor=blue,linkcolor=magenta]{hyperref}

\usepackage{tikz}
\tikzset{every picture/.style=remember picture}

\newcommand{\poly}{\operatorname{poly}}

\newcommand{\AC}{\mathcal{A}}
\newcommand{\BC}{\mathcal{B}}
\newcommand{\CC}{\mathcal{C}}

\newcommand{\EC}{\mathcal{E}}

\newcommand{\SC}{\mathcal{S}}
\newcommand{\TC}{\mathcal{T}}

\renewcommand{\geq}{\geqslant}
\renewcommand{\leq}{\leqslant}

\newcommand*{\id}{\openone}

\newcommand{\bs}{\textsf{BS}}

\def\be{\begin{equation}}
\def\ee{\end{equation}}
\def\bs{\begin{split}}
\def\e{\end{split}}
\def\ba{\begin{eqnarray}}
\def\bea{\begin{eqnarray}}

\def\tea{\end{eqnarray}}
\def\ea{\end{eqnarray}}
\def\eea{\end{eqnarray}}

\newtheorem{theorem}{Theorem}
\newtheorem{lemma}{Lemma}
\newtheorem{corollary}{Corollary}

\newtheorem{definition}{Definition}

\usepackage{amssymb}
\usepackage{dsfont}

\def\be{\begin{equation}}
\def\te{\end{equation}}
\def\ee{\end{equation}}
\def\ba{\begin{eqnarray}}
\def\bea{\begin{eqnarray}}

\def\tea{\end{eqnarray}}
\def\ea{\end{eqnarray}}
\def\eea{\end{eqnarray}}

\newcommand{\beq}{\begin{equation}}
\newcommand{\eeq}{\end{equation}}

\begin{document}

\title{Exact spectral gaps of random one-dimensional quantum circuits}

\author{Andrew E. Deneris}
\affiliation{Information Sciences, Los Alamos National Laboratory, Los Alamos, NM 87545, USA}
\author{Pablo Bermejo}
\affiliation{Information Sciences, Los Alamos National Laboratory, Los Alamos, NM 87545, USA}
\affiliation{Donostia International Physics Center, Paseo Manuel de Lardizabal 4, E-20018 San Sebasti\'an, Spain}
\affiliation{Department of Applied Physics, Gipuzkoa School of Engineering, University of the Basque
Country (UPV/EHU), Plaza Europa 1, 20018 San Sebastián, Spain}
\author{Paolo Braccia}
\affiliation{Theoretical Division, Los Alamos National Laboratory, Los Alamos, NM 87545, USA}

\author{Lukasz Cincio}
\affiliation{Theoretical Division, Los Alamos National Laboratory, Los Alamos, NM 87545, USA}

\author{M. Cerezo}
\thanks{cerezo@lanl.gov}
\affiliation{Information Sciences, Los Alamos National Laboratory, Los Alamos, NM 87545, USA}
\affiliation{Quantum Science Center, Oak Ridge, TN 37931, USA}

\begin{abstract}
    The spectral gap of local random quantum circuits is a fundamental property that determines how close the  moments of the circuit's unitaries match those of a  Haar random distribution. When studying spectral gaps, it is common to bound these quantities using tools from statistical mechanics or via quantum information-based inequalities. By focusing on the second moment of one-dimensional unitary circuits where nearest neighboring gates act on sets of qudits (with open and closed boundary conditions), we show that one can exactly compute the associated spectral gaps. Indeed, having access to their functional form allows us to prove several important results, such as the fact that the spectral gap for closed boundary condition is exactly the square of the gap for open boundaries, as well as improve on previously known bounds for approximate design convergence. Finally, we  verify our theoretical results by numerically computing the spectral gap for systems of up to 70 qubits, as well as comparing them to gaps of random orthogonal and symplectic circuits. 
\end{abstract}

\maketitle

\section{Introduction}
Random quantum circuits have played an important role in quantum computation and quantum information sciences.  Their study has allowed researchers to find experiments capable of achieving a quantum advantage by offering classically hard-to-simulate sampling experiments~\cite{boixo2018characterizing,arute2019quantum,wu2021strong,dalzell2022randomquantum,oszmaniec2022fermion,huang2021provably}, providing insights into quantum  chaos~\cite{nahum2017quantum,von2018operator,nahum2018operator,ho2022exact} and entanglement transitions~\cite{li2018quantum,skinner2019measurement,jian2020measurement}, as well as improving the study of the trainability of variational quantum algorithms~ \cite{mcclean2018barren,cerezo2020cost,pesah2020absence,napp2022quantifying,ragone2023unified,fontana2023theadjoint,larocca2024review} and error mitigation techniques~\cite{hu2024demonstration}. 

In particular, random  circuits composed of local unitary gates  sampled independently and identically distributed (i.i.d.) according to the Haar measure over the standard representation of $\mathbb{U}(d^2)$, and  acting on neighboring pairs of qudits in a one-dimensional lattice,  have received considerable attention. Here, several works~\cite{harrow2009random,brandao2016local,hunter2019unitary,haferkamp2022random,haferkamp2021improved,brown2010random,nakata2017efficient,Haferkamp2022randomquantum,harrow2018approximate,chen2024efficient,chen2024incompressibility,belkin2023approximate,mittal2023local,schuster2024random} have provided bounds for their  spectral gap--the absolute value of the largest non-one eigenvalue of the circuit's $t$-th moment operator. Such bounds have proven to be extremely useful as they can be used to quantify the number of gates needed for the circuit to form approximate designs~\cite{dankert2009exact}, and can even lead to useful insights for circuits on generic architectures~\cite{belkin2023approximate}. Still, despite the tremendous importance of spectral gaps, their exact values are usually unknown.

In this work we contribute to the body of knowledge of random quantum circuits by exactly computing the $t=2$ spectral gap for unitary circuits composed of nearest-neighboring gates acting on sets of $m$ qudits in a line (open boundary conditions) and in a circle (closed boundary conditions). Our results show that the spectral gap for closed boundary conditions is exactly the square of the spectral gap for open boundary conditions for all $m$, as well as for any number of qudits $n$. Moreover, we prove that having access to the exact functional form of the spectral gaps allows us to provide tighter bounds on the number of layers needed for the circuit to become an approximate $2$-design. As a by-product of our results, we can also fully characterize the eigenvector associated with the spectral gap.

To finish, we also considered orthogonal and symplectic random circuits. While  we were unable to exactly compute their associated spectral gaps using our presented techniques, we numerically compute them  for systems of up-to $n=70$ qubits. Our numerical results indicate that in this case the closed and open boundary conditions are not related by a square, further evidencing that their behavior is more complex than that of unitary circuits.

\section{Framework} 

\subsection{Moment operators}
Let $\HC=(\mathbb{C}^d)^{\otimes n}$ be the Hilbert space of $n$ qudits. In this work we study random circuits $U:\HC\rightarrow \HC$ in a one-dimensional lattice composed of gates acting on alternating groups of 
 $m$ nearest-neighboring qudits (we assume for simplicity that $\eta=n/m\in\mathbb{N}$). Then, the circuit $U$ can be expressed as a product of layers of the form
\begin{equation}\label{eq:circuit}
    U=\prod_{l=1}^L U_{l}\,,
\end{equation}
where 
\begin{equation}\label{eq:layer}
 U_{l}=\left( \prod_{j = 1}^{\frac{\eta}{2}} U_{2j-1,2j}^l\right)(U_{\eta,1}^l)^\Delta \left(\prod_{j = 1}^{\frac{\eta-2}{2}} U_{2j,2j+1}^l\right)\,.
\end{equation}
Here, $U_{j,j'}$ is a local unitary that acts non-trivially on the  $j$-th and $j'$-th groups of neighboring $m$ qudits and trivially on the rest. Then,  $\Delta =0$ represents open boundary conditions (unitaries acting in a line), while $\Delta =1$ denotes closed boundary conditions (unitaries acting in a circle). We show such a circuit in Fig.~\ref{fig:circ}. 

We then assume that each $U_{j,j'}$ is sampled i.i.d. according to the Haar measure over some local group $G_{j,j'}^l$. Specifically, we will consider the cases when $G_{j,j'}^l=\mathbb{U}(d^{2m})$  $\mathbb{O}(d^{2m})$ or $\mathbb{SP}(d^{2m}/2)$ (assuming $d$ even). Crucially, it is known that in the large number of layers limit (i.e., as $L\rightarrow \infty$), the distribution of unitaries will converge to that of some global group $\mathbb{G}\subseteq \mathbb{U}(d^n)$ that is determined by the local groups from which the gates are sampled. 

In particular, we will focus on  the question: \textit{How fast does the distribution of unitaries $\EC_L$ obtained from a shallow circuit $U$ converge to being an approximate $2$-design over $\mathbb{G}$?} To answer this, it is convenient to define the second moment operators
\small
\begin{align}
    \TC_\mathbb{G}^{(2)}&=  \underset{U \sim \mathbb{G}}{\mathbb{E}}[U^{\otimes 2} \otimes (U^*)^{\otimes 2}] = \int_\mathbb{G}d\mu \,U^{\otimes 2} \otimes (U^*)^{\otimes 2}\,, \label{eq:psi-r}\\ 
    \TC_{\EC_L}^{(2)}&= \underset{U \sim \EC_L}{\mathbb{E}}[U^{\otimes 2} \otimes (U^*)^{\otimes 2}] = \int_{\EC_L}dU U^{\otimes 2}\otimes (U^*)^{\otimes 2} \label{eq:tiwrlG}\,,
\end{align}
\normalsize
as well as their difference
\begin{equation}
    \AC_L^{(2)}=\TC_\mathbb{G}^{(2)}- \TC_{\EC_L}^{(2)}\label{eq:A}\,.
\end{equation}
Above,  $\int_\mathbb{G}d\mu$ denotes the Haar measure over $\mathbb{G}$, and  $\int_{\EC_L}dU$ the average over the set of unitaries $\EC_L$ (and the associated distribution $dU$)  obtained from the  $L$-layered  circuit $U$ as in Eq.~\eqref{eq:circuit}. 

\begin{figure}[t]
    \centering
    \includegraphics[width=.9\columnwidth]{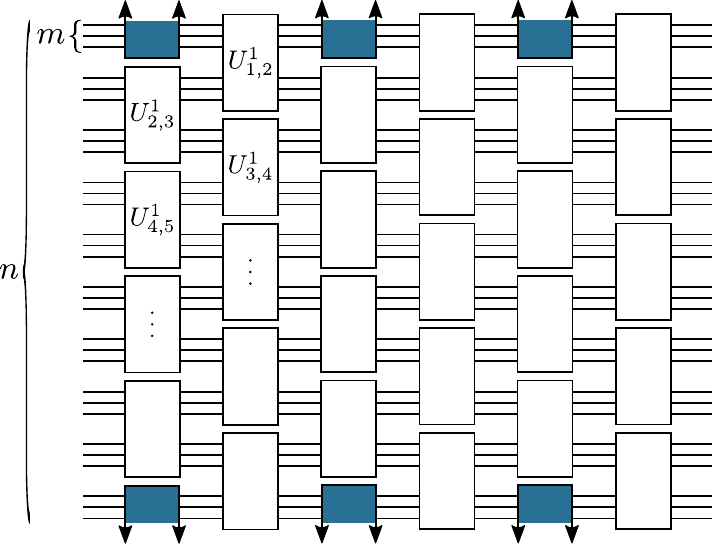}
    \caption{\textbf{Schematic representation of a one-dimensional random circuit.} As shown, random gates act on alternating groups of $m$ neighboring qudits in a brick-like fashion. The  colored gates are removed (added) with open (closed) boundary conditions.}
    \label{fig:circ}
\end{figure}

Due to the fact that the distribution of unitaries will converge to that of some global group as $L \rightarrow \infty$ we can therefore rephrase our question of interest as: \textit{How many layers $L$ are required for $\TC_{\EC_L}^{(2)}$ to be $\varepsilon$ close to $\TC_{\mathbb{G}}^{(2)}$}. As such, the matrix $\AC_L^{(2)}$ can be used to quantify how much the $2$-nd moments of the distributions differ. In particular, we will focus on the following definition.
\begin{definition}[Approximate design.]\label{def:design}
    We say that $\EC_L$ forms an $\varepsilon$-approximate $2$-design if $|\AC_L^{( t)}|_{\infty}\leq\varepsilon/d^n$. 
\end{definition}
We note that while we here focus on the infinity norm of $\AC_L^{(2)}$, known as the monomial definition of an approximate $t$-design, this is not the only way to quantify the closeness between two ensembles of unitaries. For instance, one can also use other norms, such as the operator diamond norm which bounds difference in measurements thanks to its operational meaning within the context of channel discrimination~\cite{nielsen2000quantum}. In addition, recent works have also focused on the  stronger notion of a design as measured via the relative error rather (i.e.,  checking wether $(1-\varepsilon)\TC_\mathbb{G}^{(2)}\preceq \TC_{\EC_L}^{(2)}\preceq(1+\varepsilon)\TC_\mathbb{G}^{(2)}$, where $A \preceq B$ if and only if $B - A$ is a completely positive map)~\cite{schuster2024random}. We decide to use the  monomial measure of Definition~\ref{def:design} as its evaluation is direct from the characterization of the eigenvalues of the moment operator, our main goal. Moreover, we note that once the spectral gap is computed, and the eigenvalues are known, then such knowledge can be used to evaluate other norms or definitions of designs. We leave this research director for future work.

Given that $\TC_{\mathbb{G}}^{(2)}$ represents the second moment of a unitary uniformly sampled from a group, then we can leverage the tools from Weingarten calculus (see the next section and also Ref.~\cite{collins2006integration,puchala2017symbolic,mele2023introduction}) to analytically evaluate this operator. However, the analysis of $\TC_{\EC_L}^{(2)}$ is not as straightforward, as the unitaries in $\EC_L$ might not form a group, or since  $dU$ might not be a Haar measure. As such, Weingarten calculus cannot be directly used to study $\TC_{\EC_L}^{(2)}$.  However, we can simplify the evaluation of $\TC_{\EC_L}^{(2)}$ as follows. From Eq.~\eqref{eq:circuit}, we can define $\TC_{\EC_L}^{(2)} = \prod_{l = 1}^L \TC_{\EC_l}^{(2)}$ where $\TC_{\EC_l}^{(2)}$ is the second moment operator over a single layer, $U_l$, of the random circuit. Then, from Eq.~\eqref{eq:layer} we find
\begin{equation}\label{eq:layer-moment}
 \TC_{\EC_l}^{(2)}=\left( \prod_{j = 1}^{\frac{\eta}{2}} \TC^{(2)}_{G_{2j-1,2j}^l}\right) (\TC_{G_{\eta,1}^{l}}^{(2)})^\Delta  
 \left(\prod_{j = 1}^{\frac{\eta-2}{2}} \TC^{(2)}_{G_{2j,2j+1}^l}\right)\,.
\end{equation}
 Here $\TC_{G_{j,j'}^{l}}^{(2)}$ represents the second moment operator of a single random local gate that is sampled according to the Haar measure over  $G_{j,j'}^l$. In particular, if one assumes  that $ G_{j, j'}^{l}$ is the same  for all $j, j'$ and $l$, then all the local moment operators are equal. More generally however, if the local gates are sampled from different groups, but they are constant across layers, i.e., $G_{j, j'}^{l}$ is the same  for all $l$, then $\TC_{\EC_l}^{(2)}=\TC_{\EC_l'}^{(2)}$ for all $l,l'$. In this case, we have
\begin{equation}\label{eq:layer-circuit}
    \TC_{\EC_L}^{(2)} = (\TC_{\EC_1}^{(2)})^L\,.
\end{equation}

Indeed from the previous discussion we find that the study of $\TC_{\EC_L}^{(2)}$ can be reduced to that of $\TC_{\EC_l}^{(2)}$, which in turn requires the evaluation of the local operators $\TC_{G_{j,j'}^{l}}^{(2)}$. Given that these correspond to the moments of a unitary sampled from the Haar measure over a group, we can study them via Weingarten calculus presented in the next section.

\subsection{Weingarten calculus}

Weingarten calculus is a mathematical tool used to compute averages of functions of unitaries---and their complex conjugates---when the unitaries are randomly sampled according to the Haar measure, $d\mu$, over a group $G$. We refer the reader to~\cite{mele2023introduction} for an introduction to this technique, but for completeness we review here some basic concepts.

First, we define as ${\rm comm}^{(2)}(G)$ the second order commutant of $G$, i.e., the vector space of all matrices commuting with the two-fold tensor product of the unitaries in $G$. Mathematically, this means

\begin{equation}
    {\rm comm}^{(2)}(G)=\{A\in\BC(\HC^{\otimes 2})\,|\, [A,U^{\otimes 2}]=0\,,\forall U\in G\}\,,\nonumber
\end{equation}

where $\BC(\HC^{\otimes 2})$ denotes the set of bounded operators acting on the Hilbert space $\HC^{\otimes 2}$. Since $\TC_G^{(2)}$ projects into ${\rm comm}^{(2)}(G)$~\cite{mele2023introduction}, then given a basis $\{ P_{\mu} \}_{\mu = 1}^{D}$ of ${\rm comm}_G^{(2)}$ one can readily find that
\begin{equation}
    \TC_G^{(t)} = \sum_{\mu, \nu = 1}^{D} (W_{G}^{-1})_{\mu \nu}\ket{P_{\mu}}\rangle\langle \bra{P_{\nu}}\,,
\end{equation}
where $W_{G}$ is the Gram Matrix of the commutant's basis under the Hilbert-Schmidt inner product. That is,  $(W_{G}^{-1})_{\mu \nu}= \Tr[P_{\mu}^\dag P_{\nu}]$. Above, $\ket{A}\rangle$ denotes the vectorization of an operator $A\in \BC(\HC^{\otimes 2})$, such that if $A=\sum_{i,j}a_{ij}\ket{i}\bra{j}$, then $\ket{A}\rangle=\sum_{i,j}a_{ij}\ket{i}\ket{j}$. Then, the inner product of two vectorized operators is $\langle \langle B|A\rangle\rangle=\Tr[B\ad A]$.

\subsection{Spectral gap}

From the discussion in the previous section we know that $\TC_\mathbb{G}^{(2)}$ is a projector onto the second order commutant of the group $\mathbb{G}$, meaning that it can only have eigenvalues that are equal to one, or equal to zero. That is, the eigenvectors of $\TC_\mathbb{G}^{(2)}$ can be divided into two subspaces 
\begin{itemize}
    \item $\SC_\mathbb{G}(1)=\{\ket{v}\,|\,\TC_\mathbb{G}^{(2)}\ket{v}=\ket{v}\}$,
    \item $\SC_\mathbb{G}(0)=\{\ket{v}\,|\,\TC_\mathbb{G}^{(2)}\ket{v}=0\}$.
\end{itemize}
Moreover, we can use Eq.~\eqref{eq:layer-circuit} to note that while $\TC_{\EC_1}^{(2)}$ is not a projector on its own, it must have eigenvalues equal to one, eigenvalues strictly smaller than one, and eigenvalues equal to zero. That is, 
the eigenvectors of $\TC_{\EC_1}^{(2)}$ can be divided into three subspaces
\begin{itemize}
    \item $\SC_{\EC_1}(1)=\{\ket{v}\,|\,\TC_{\EC_1}^{(2)}\ket{v}=\ket{v}\}$,
    \item  $\SC_{\EC_1}(\lambda)=\{\ket{v}\,|\,\TC_{\EC_1}^{(2)}\ket{v}=\lambda\ket{v}\}$, with $0<\lambda<1$,
   \item  $\SC_{\EC_1}(0)=\{\ket{v}\,|\,\TC_{\EC_1}^{(2)}\ket{v}=0\}$.
\end{itemize}
Importantly, $\TC_\mathbb{G}^{(2)}$ and $\TC_{\EC_1}^{(2)}$ share exactly the same number of eigenvalues equal to one, as well as the associated eigenvectors. That is $\SC_\mathbb{G}(1)=\SC_{\EC_1}(1)$.  This is due to the fact that in the limit $L\rightarrow \infty$ one must recover (by assumption) $\TC_{\EC_L}^{(2)}=(\TC_{\EC_1}^{(2)})^L=\TC_\mathbb{G}^{(2)}$ . Hence, combining Definition~\ref{def:design} with Eqs.~\eqref{eq:A} and~\eqref{eq:layer-circuit} shows that 
\begin{equation}\label{eq:spectral-gap}
    |\AC_L^{( t)}|_{\infty}=(\lambda_d^{(n,m)})^L
\end{equation}
where $\lambda_d^{(n,m)}$ is known as the spectral gap, or the largest non-one eigenvalue of $\TC_{\EC_1}^{(2)}$. That is,
\begin{equation}
    \lambda_d^{(n,m)}=\max_{\ket{v}\in \SC_{\EC_1}(\lambda) } \frac{\bra{v}\TC_{\EC_1}^{(2)}\ket{v}}{\langle v|v\rangle}\,.
\end{equation}
In what follows we use Eq.~\eqref{eq:spectral-gap} to evaluate the spectral gap through a study of the eigenvalues of $\TC_{\EC_1}^{(2)}$.

\section{Theoretical results for unitary random circuits}

Our main theoretical result is an exact formula for the spectral gap when $G_{j,j'}^l=\mathbb{U}(d^{2m})$ $\forall j,j',l$, and thus when $\mathbb{G}=\mathbb{U}(d^n)$ for both open and closed boundary conditions. In addition, we also provide in the next section an analytical expression for the associated eigenvector.

In particular, we find that the following theorem holds:
\begin{theorem}\label{th:theorem_1}
Let $U$ be a circuit as in Eqs.~\eqref{eq:circuit} and~\eqref{eq:layer}, i.e.,
\begin{equation}
    U=\prod_{l=1}^L U_{l}\,,\nonumber
\end{equation}
and
\begin{equation}
 U_{l}=\left( \prod_{j = 1}^{\frac{\eta}{2}} U_{2j-1,2j}^l\right)(U_{\eta,1}^l)^\Delta \left(\prod_{j = 1}^{\frac{\eta-2}{2}} U_{2j,2j+1}^l\right)\,.\nonumber
\end{equation}
where all the local gates are  sampled i.i.d from the Haar measure over $\mathbb{U}(d^{2m})$. Then, for open boundary conditions ($\Delta=0$) the spectral gap is
    \begin{equation}\label{eq:lambda_max_open_th}
    \lambda_{d,{\rm open}}^{(n,m)}= \left(\frac{d^m}{d^{2m}+1}\right)^2\left(2+2\cos\left(\frac{2m\pi}{n}\right)\right)\,,
\end{equation}
while for closed boundary conditions ($\Delta=1$) it is
\begin{align}\label{eq:lambda_max_closed_th}
    \lambda_{d,{\rm closed}}^{(n,m)}&= \left(\frac{d^m}{d^{2m}+1}\right)^4\left(2+2\cos\left(\frac{2m\pi}{n}\right)\right)^2\\&=\left(\lambda_{d,{\rm open}}^{(n,m)}\right)^2\,.\nonumber
\end{align}
\end{theorem}

As we now discuss, Theorem~\ref{th:theorem_1} has several important implications. First, we can see that the spectral gap for closed boundary conditions is exactly the square of that for open boundary conditions. Next, from Eqs.~\eqref{eq:lambda_max_open_th} and~\eqref{eq:lambda_max_closed_th} we can quantify the number of layers needed for the circuit to be an approximate $2$-design over $\mathbb{G}=\mathbb{U}(d^n)$ as follows.
\begin{corollary}\label{cor:corollary_1}
    The circuit $U$ will become an $\varepsilon$-approximate $2$-design for open boundary conditions when the number of layers is
\begin{equation}\label{eq:L-closed-cor}
    L\geq \frac{\left(n\log(d)+\log(\frac{1}{\varepsilon})\right)}{\left(2\log(\frac{d^{2m}+1}{d^m})-\log\left(2+2\cos\left(\frac{2m\pi}{n}\right)\right)\right)}\,.\nonumber
\end{equation}
\normalsize
while for closed boundary conditions when 
\begin{equation}\label{eq:L-closed-cor-closed}
    L\geq \frac{\left(n\log(d)+\log(\frac{1}{\varepsilon})\right)}{\left(4\log(\frac{d^{2m}+1}{d^m})-2\log\left(2+2\cos\left(\frac{2m\pi}{n}\right)\right)\right)}\,.\nonumber
\end{equation}
\normalsize
\end{corollary}
Hence, in the large $n$ limit Corollary~\eqref{cor:corollary_1} implies 
\begin{equation}\label{eq:bound-ours}
    L\geq \frac{1}{C}\left(n\log(d)+\log(\frac{1}{\varepsilon})\right)\,,
\end{equation}
where $C=2\log(\frac{d^{2m}+1}{2d^m})$ for open boundary conditions and $C=4\log(\frac{d^{2m}+1}{2d^m})$ for closed boundaries. That is, when we have closed boundary conditions, we need half of the layers for $U$ to become an $\varepsilon$-approximate $2$-design.

Importantly, for the case of $m=1$ we can compare the bound in Eq.~\eqref{eq:bound-ours} with that of Ref.~\cite{hunter2019unitary}, which states that for open boundaries the circuit is an $\varepsilon$-approximate $2$-design if 
\begin{equation}\label{eq:bound-hunter}
    L\geq \frac{1}{\widetilde{C}}\left(2n\log(d)+\log(n)+\log(\frac{1}{\varepsilon})\right)\,,
\end{equation}
 where $\widetilde{C}=\log(\frac{d^{2m}+1}{2d^m})$. Hence, comparing Eqs.~\eqref{eq:bound-ours} and~\eqref{eq:bound-hunter} shows that our bound drops a $\log(n)$ term in the numerator and improves on the constants. In particular, our results improve the linear term from $\approx 6.2 n$ to $\approx1.55 n$ when $d=2$ (and to $\approx0.775 n$ in closed boundary conditions). Similarly, our result improves the asymptotic convergence from $2n$ to $n/2$ when $d\rightarrow \infty$ (and to $n/4$ in closed boundary conditions).  

Here it is worth recalling that while $L$ quantifies the number of layers, this does not correspond to the number of gates in the circuit. Indeed, one still needs to account for the number of gates $N$ needed to make each  gate $G_{j,j'}^{l}$ be Haar random over $\mathbb{U}(d^{2m})$. Using known constructions for approximate $2$-designs~\cite{chen2024incompressibility} or pseudorandom unitaries~\cite{metger2024simple,schuster2024random} one can build such local unitaries  with $\poly(m)$ gates (we will assume quadratic scaling~\cite{metger2024simple}). With this result in mind, we can now use Eq.~\eqref{eq:bound-ours} to compare the total number of gates depth needed for $U$ to be an approximate $2$-design as a function of $m$. In particular, one can find that there is a positive trade-off  when $m>1$, and even when it is taken to be  a function of  $n$. For instance, let us compare the cases when $m=1$ and $m=\log(n)$. Here,  the associated total number of gates, denoted as  $N_1$ and $N_{\log(n)}$, respectively, satisfy the asymptotic ratio (valid for both open and closed boundary conditions)
\begin{equation}
    \frac{N_1}{N_{\log(n)}}\xrightarrow[n\rightarrow \infty]{}\frac{\log(2)}{\log(\frac{5}{4})}\approx 3.1\,,
\end{equation}
which indicates that grouping qudits into sets of $\log(n)$ leads to a more favorable asymptotic scaling.

\begin{figure*}[t]
    \centering
    \includegraphics[width=1\linewidth]{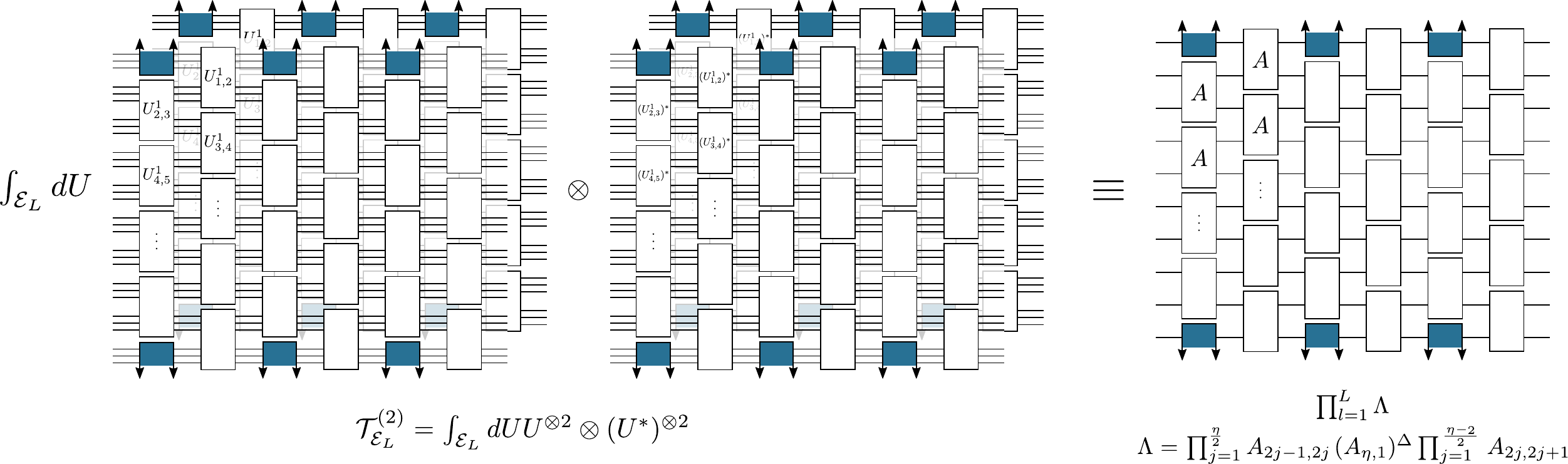}
    \caption{\textbf{Schematic representation of the second moment operator.} After the reduction of Eq.~\eqref{eq:layer-moment}, the action of $\TC_{\EC_L}^{(2)}$ can be studied as a $2^\eta\times 2^\eta$ operator where $A$ matrices act on two two-dimensional spaces following the same topology as that in $U$.}
    \label{fig:redution}
\end{figure*}

\section{Proof of  Theorem~\ref{th:theorem_1} and explicit eigenvector construction}

In what follows, we present the proof of  Theorem~\ref{th:theorem_1}. While in general we tend to present proofs in the appendices, we believe our construction possesses useful insights that are worth highlighting in the main text. In particular, the notation introduced here will be important to understand the form of the eigenvectors of $\TC_{\EC_1}^{(2)}$ associated with the eigenvalues in Eqs.~\eqref{eq:lambda_max_open_th} and~\eqref{eq:lambda_max_closed_th}. As we will see, our proof technique consists of a series of dimensionality reductions and simplifications.

We start by finding the vectorized moment operator over a single local random gate $\TC_{G_{j,j'}^{l}}^{(2)}$ for the case when $G_{j,j'}^l=\mathbb{U}(d^{2m})$ for all $j, j',l$. We refer the reader to ~\cite{dalzell2021random,braccia2024computing} for additional details on this first step. To begin, we note that $\TC_{\mathbb{U}(d^{2m})}^{(2)}$ acts on $\HC_A\otimes \HC_B$, where  $\HC_A=\HC_{A,1}\otimes \HC_{A,2}$ and $\HC_{A,1}=\HC_{A,2}=(\mathbb{C}^{d})^{\otimes m}$ (a similar expression holds for $\HC_B$). For convenience, we will re-write this Hilbert space as $\HC_{A,1}\otimes\HC_{B,1}\otimes \HC_{A,2}\otimes\HC_{B,2}$. Then, since $\TC_{\mathbb{U}(d^{2m})}^{(2)}$ is a moment operator of a unitary sampled over the Haar measure over $\mathbb{U}(d^{2m})$  it projects into the second order commutant  ${\rm comm}^{(2)}(\mathbb{U}(d^{2m})) = {\rm span}_{\CC} \{ \id_1 \otimes \id_2, {\rm SWAP}_1\otimes {\rm SWAP}_2  \}$, where $\id_i$ denotes the identity over $\HC_{A,i}\otimes\HC_{B,i}$ and ${\rm SWAP}_i$ denotes the swap  operator such that  ${\rm SWAP}\ket{\psi}\otimes \ket{\phi}= \ket{\phi}\otimes\ket{\psi}$ for $\ket{\psi},\ket{\phi}\in (\mathbb{C}^{d})^{\otimes m}$. From the previous, we can express the non-trivial action of $\TC_{\mathbb{U}(d^{2m})}^{(2)}$ on the vector space spanned by the basis  $\{ \ket{i}, \ket{s} \}^{\otimes 2}$ obtained from the vectorization
\begin{equation}
    \ket{i} \equiv |\id\rangle\rangle\,, \quad \ket{s} \equiv |{\rm SWAP}\rangle\rangle\,.
\end{equation}

Then, from the definition of the second moment operator we find that its action on this four-dimensional basis is given by~\cite{braccia2024computing,garcia2023deep}
\begin{align}
    \TC_{\mathbb{U}(d^{2m})}^{(2)} \ket{ii} &= \ket{ii} \,,\nonumber\\
    \TC_{\mathbb{U}(d^{2m})}^{(2)} \ket{is} &= \frac{d^m}{d^{2m}+1} \ket{ii} + \frac{d^m}{d^{2m}+1} \ket{ss} \,,\nonumber\\
    \TC_{\mathbb{U}(d^{2m})}^{(2)} \ket{si} &= \frac{d^m}{d^{2m}+1} \ket{ii} + \frac{d^m}{d^{2m}+1} \ket{ss}\,,\nonumber\\
    \TC_{\mathbb{U}(d^{2m})}^{(2)} \ket{ss} &= \ket{ss} \,.\nonumber
\end{align}
Therefore, we can represent the action of the local second moment operator, $\TC_{\mathbb{U}(d^{2m})}^{(2)}$, as the $4 \times 4$ matrix $A$
\begin{equation}\label{eq:matrix-P} A = \begin{pmatrix}
        1 & \frac{d^m}{d^{2m}+1} & \frac{d^m}{d^{2m}+1} & 0\\
        0 & 0 & 0 & 0 \\
        0 & 0 & 0 & 0 \\
        0 & \frac{d^m}{d^{2m}+1} & \frac{d^m}{d^{2m}+1} & 1
    \end{pmatrix}\,,
\end{equation}
that acts on $(\mathbb{R}^{2})^{\otimes 2}$ spanned by $\{ \ket{i}, \ket{s} \}^{\otimes 2}$. Indeed, one  can easily verify that this matrix is a projector since $A^2 = A$, as expected.

From here, we can use Eq.~\eqref{eq:layer-moment} to represent the single-layer moment operator $\TC_{\EC_1}^{(2)}$ by the $2^\eta\times 2^\eta$ operator 
\begin{equation}\label{eq:layer-moment-lambda}
 \Lambda=\left(\prod_{j = 1}^{\frac{\eta}{2}} A_{2j-1,2j}\right)(A_{\eta,1})^\Delta  \left(\prod_{j = 1}^{\frac{\eta-2}{2}} A_{2j,2j+1} \right)\,.
\end{equation}
Above, $A_{j,j'}$ denotes that an $A$ matrix acts on the $j$-th and $j'$-th dimensional copies of $\mathbb{R}^{2}$. Then, as shown in Fig.~\ref{fig:redution}, $\TC_{\EC_L}^{(2)}$ can be represented by $\prod_{l=1}^L\Lambda = \Lambda^L$, which can be visualized  as a product of $A$ matrices that satisfy the same topology as that in $U$.

Before proceeding to the next section, let us highlight the fact that $\Lambda$ acts on the vector space $(\mathbb{R}^{2})^{\otimes \eta}$ with basis $\{ \ket{i}, \ket{s} \}^{\otimes \eta}$,  which we henceforth represent in the form $|\alpha_1 \alpha_2 \alpha_3 \dots \alpha_{\eta} \rangle$ where $\alpha_j \in \{ i, s \} \quad \forall{j}$.  For convenience we can also write them as $\ket{\alpha_0}^{\otimes k_0} \otimes \ket{\alpha_1}^{\otimes k_1} \otimes \dots \otimes \ket{\alpha_{\zeta}}^{\otimes k_\zeta}$ where $\sum_{j = 0}^\zeta k_j = \eta$. This divides a state into blocks of $i$'s and $s$'s. Here we define $\zeta$ as the number of ``\textit{switches}'' between neighboring blocks of $i$ and $s$.

From the  previous we see that  evaluating the spectral gap is equivalent to  computing the largest non-one eigenvalue of $\Lambda$, which we now proceed to do for both open and closed boundary conditions.

\subsubsection{Open boundary conditions}

Starting from Eq.~\eqref{eq:layer-moment-lambda}, and taking  $\Delta = 0$, we find that we can write
\begin{equation}\label{eq:lambda-simple}
    \Lambda=A^{\otimes \frac{\eta}{2}}\left(\id_2 \otimes A^{\otimes \frac{\eta - 2}{2}} \otimes \id_2\right)\,,
\end{equation}
where $\id_2$ is a $2\times 2$ identity matrix. Then, we can show that the following two important Lemmas hold (see the appendix for a proof).

\begin{lemma}\label{lemma:1}
    An eigenvector of $\Lambda$ must also be an eigenvector of $A^{\otimes \frac{\eta}{2}}$ and thus can only be written in terms of basis vectors $\ket{\alpha_0}^{\otimes k_0} \otimes \ket{\alpha_1}^{\otimes k_1} \otimes \dots \otimes \ket{\alpha_\zeta}^{\otimes k_\zeta}$ such that $k_j$ is even $\forall j$.
\end{lemma} 

\begin{lemma}\label{lemma:2}
If a vector is a linear combination of basis terms with $\zeta$, or less, switches then the action on this vector by $\Lambda$ will result in a linear combination of basis states with $\zeta$, or less, switches. 
\end{lemma}

First, let us note that Lemma~\ref{lemma:1} implies that the eigenvectors of $\Lambda$ can only have support in a subspace of $(\mathbb{R}^2)^{\otimes \eta}$ containing basis states composed of blocks of even number of $i$'s and $s$'s. Thus, we will henceforth only consider the action of $\Lambda$ in such a subspace. Next, Lemma~\ref{lemma:2} shows that the action of $\Lambda$ cannot increase the number of switches between $i$'s and $s$'s. These results allow us to reorder the rows and columns of $\Lambda$ according to the number of switches they have. After such reordering, we can express the action of $\Lambda$ as an upper triangular matrix of the form
\begin{equation}\label{eq:matrix}
    \Lambda \equiv\begin{pmatrix}
  \fbox{$\begin{matrix} 1 & 0\\ 0 & 1  \end{matrix}$}  &  \cdots   & \cdots& \cdots   \\
    {\textrm{\Large 0}}  & \fbox{$\begin{matrix} B_1 \end{matrix}$} &  \cdots & \cdots \\
     {\textrm{\Large 0}}  &   {\textrm{\Large 0}}   &  \fbox{$\begin{matrix}  & & &\\ & B_2 &\\
   & & &\end{matrix}$}  & \cdots\\
          \textrm{\Large 0}     &   {\textrm{\Large 0}}   & {\textrm{\Large 0}} &   \ddots         & 
  \end{pmatrix}\,.
\end{equation}
Here we have defined the square matrices $B_\zeta$ as the submatrices of $\Lambda$ that act on basis states containing exactly $\zeta$-switches. Indeed,  the fact that the lower half of this representation of  $\Lambda$ is zero  follows from  Lemma~\ref{lemma:2}: $\Lambda$ does not increase the number of  switches. 

Next, we will use the well known property stating that the eigenvalues of an upper triangular block-matrix are the eigenvalues of its diagonal blocks, meaning that the eigenvalues of $\Lambda$ are precisely those of the  $B_\zeta$. Here, we can directly visualize from Eq.~\eqref{eq:matrix} that $\Lambda$ contains two eigenvalues equal to one in the $B_0$ block, whose associated eigenvectors are the  states $\ket{i}^{\otimes \eta}$ and $\ket{s}^{\otimes \eta}$. In fact, one can verify that $\SC_{\mathbb{U}(d^n)}(1)=\SC_{\EC_1}(1)$, as these are the only two eigenvectors with eigenvalue equal to one for $\TC_{\mathbb{U}(d^n)}^{(2)}$~\cite{braccia2024computing}. 

The previous realization shows that the spectral gap will be the largest eigenvalue of the $B_\zeta$ sub-matrices with $\zeta\geq1$. In what follows, we will show that such an eigenvalue is contained in $B_1$. Our proof procedure contains two steps: (1) Compute an upper bound for the eigenvalues of $B_\zeta$; (2) Explicitly evaluate the largest eigenvalue of $B_1$; (3) Show that such an eigenvalue of $B_1$ is larger than the upper bounds obtained for the eigenvalues of the $B_\zeta$ with $\zeta\geq2$, and thus that it corresponds to the spectral gap. The first step can be accomplished by the following lemma, proved in the appendix. 
\begin{lemma}\label{lemma:3}
    Let $\{\lambda_{\zeta,i} \}$ be the  eigenvalues of $B_\zeta$. Then, the following  upper  bound holds for all $i$, and $\zeta \geq 1$
\begin{equation}\label{eq:bounds}
    \lambda_{\zeta,i}\leq \left(\frac{2d^m}{d^{2m}+1}\right)^{2\zeta} \,.
\end{equation}
\end{lemma}

Next, to obtain the eigenvalues of $B_1$ we will explicitly construct this matrix by analyzing the action of $\Lambda$ on basis states with exactly one switch. To simplify the notation let us define the zero- and one-switch states as 
\begin{equation}
\begin{split}
        \ket{k}_+&=\ket{i}^{\otimes k} \otimes \ket{s}^{\otimes \eta - k}\,,\\
    \ket{k}_-&=\ket{s}^{\otimes \eta - k} \otimes \ket{i}^{\otimes k}\,,
\end{split}
\end{equation}
where we further assume that $k$ is even as per Lemma~\ref{lemma:1}. As such, we note that the bases $\{\ket{k}_\pm\}_{k=0,2,4,\ldots,\eta }$ are each of size $\frac{\eta}{2}+1$ (where the ``$+1$'' comes from the fact that we start counting at zero). 
Explicitly using Eqs.~\eqref{eq:matrix-P} and~\eqref{eq:lambda-simple} allows us to find that the action of $\Lambda$ over these states  is (see the appendix for a proof)
\begin{align}
    &\Lambda\ket{k}_{\pm} = \left(\frac{d^m}{d^{2m}+1}\right)^{2}\!(\ket{k +2}_{\pm} + 2\ket{k}_{\pm}) + \ket{k - 2}_{\pm})\,,\nonumber\\
    &\Lambda\ket{0} = \ket{0}\,,\label{eq:Lambda-1-switch}\\
    &\Lambda\ket{\eta} = \ket{\eta}\,.   \nonumber
\end{align}

Equation~\eqref{eq:Lambda-1-switch} shows that the action of $\Lambda$ does not mix the $\pm$ subspaces  spanned by $\{\ket{k}_\pm\}_{k=2,4,\ldots,\eta-2}$, indicating that $B_1$ has an internal block-diagonal structure of the form
\begin{equation}\label{eq:matrix-B1}
    B_1 =\begin{pmatrix}
  B_1^+  &  0   \\
  0 &   B_1^-  
  \end{pmatrix}\,,
\end{equation}
and thus that its eigenvalues will be two-fold degenerate. Importantly, we can use Eq.~\eqref{eq:Lambda-1-switch} to find that  $B_1^\pm$ are given by the  tridiagonal toeplitz matrix of size $\frac{\eta-2}{2} \times \frac{\eta-2}{2}$ 

\begin{equation}\label{eq:matrix-TT}
B_1^\pm = \left(\frac{d^m}{d^{2m}+1}\right)^{2}\begin{pmatrix}
2 & 1 & 0 & 0 & \dots & 0 \\
1 & 2 & 1 & 0 & \dots & 0 \\
0 & 1 & 2 & 1 & \dots & 0 \\
\vdots & & \ddots & \ddots & \ddots & \vdots\\
0 & \dots & 0 & 1 & 2 & 1\\
0 & \dots & 0 & 0 & 1 & 2\\
\end{pmatrix} \,.
\end{equation}

From here, we can use the fact that tridiagonal Toeplitz matrices have very unique and well-known properties, such as the analytical form of their eigenvalues and eigenvectors being exactly known~\cite{noschese2019eigenvector}. Thus, we we can find that the largest eigenvalue of $B_1$ is given by
\begin{align}\label{eq:lambda_max}
    \lambda^{{\rm max}}&= \left(\frac{d^m}{d^{2m}+1}\right)^{2}\left(2+2\cos\left(\frac{2\pi}{\eta}\right)\right)\nonumber\\
    &=\left(\frac{d^m}{d^{2m}+1}\right)^2\left(2+2\cos\left(\frac{2m\pi}{n}\right)\right)\,.
\end{align}
Furthermore we find that the two eigenvectors of $\Lambda$ (and thus of $\TC_{\EC_1}^{(2)}$) associated with $\lambda^{{\rm max}}$ are
\begin{align}\label{eq:x_max}
    \ket{\lambda_{\pm}^{\rm max}} =& \sin(\frac{2\pi}{\eta})\ket{2}_{\pm} + \sin(\frac{4\pi}{\eta})\ket{4}_{\pm} + \ldots\nonumber\\
    &+ \sin(\frac{(\eta-2)\pi}{\eta})\ket{\eta-2}_{\pm}\,.
\end{align}

Finally, let us compare Eqs.~\eqref{eq:bounds} and ~\eqref{eq:lambda_max}. We want to show that for $\zeta\geq 2$
\begin{equation}
  \lambda_{\zeta,i}\leq \left(\frac{2d^m}{d^{2m}+1}\right)^{2\zeta}  \leq \lambda^{{\rm max}}\,,
\end{equation}
meaning that the largest eigenvalue of $\Lambda$ is $\lambda^{{\rm max}}$. In the case of $m=1$ the inequality holds for $\eta\geq 6$. In the case of $m\geq 2$ the inequality holds for all $\eta\geq 4$. When these inequalities do not hold (i.e., for $m=1, \eta = 4$) we can explicitly construct and diagonalize the matrix in Eq.~\eqref{eq:matrix} and verify that the largest eigenvalue of $B_1$ is indeed the largest eigenvalue of $\Lambda$. Hence, we can find that the spectral gap for open boundary conditions is 
\begin{equation}
   \lambda_{d,{\rm open}}^{(n,m)}= \lambda^{{\rm max}}\,.
\end{equation}

\subsubsection{Closed boundary conditions}

Next, we consider the case of closed boundary conditions, i.e., $\Delta=1$ in Eqs.~\eqref{eq:layer} and~\eqref{eq:layer-moment-lambda}. A key difference here is that a term like $\ket{iiiiiiss}$ actually has 2 switches instead of 1 and is thus equivalent to terms $\ket{iiiissii}$, $\ket{iissiiii}$ and $\ket{ssiiiiii}$ based on the action of $\Lambda$ up to some cyclical permutation. In general, let $ C_{\frac{\eta}{2}}$ be the cyclical group of order $\frac{\eta}{2}$, given by the integer power of an element $\sigma$ such that the action of sigma is 
\begin{equation}
    \sigma |\alpha_1 \alpha_2 \alpha_3\alpha_4 \dots \alpha_\eta \rangle=|\alpha_{\eta-1} \alpha_\eta \alpha_1 \alpha_2 \dots \alpha_{\eta-2} \rangle\,.
\end{equation}
That is, $\sigma$ performs a two-site translation. Importantly, one can readily verify that the action of sigma commutes with that of $\Lambda$, i.e., $[\Lambda, \sigma]=0$. Due to this symmetry, basis vectors are uniquely determined by the number of $i$'s (or $s$'s) they possess, and not by their relative position. As such,  there is no need to distinguish between $\ket{k}_+$ and $\ket{k}_-$ states as they are related by a cyclical permutation. 

Similar to the case of open boundary conditions,  we find that the action of $\Lambda$  can be expressed as an upper triangular block-matrix on a basis labeled by the number of switches as  
\begin{equation}\label{eq:matrix-closed}
    \Lambda \equiv\begin{pmatrix}
  \fbox{$\begin{matrix} 1 & 0\\ 0 & 1  \end{matrix}$}  &  \cdots   & \cdots& \cdots   \\
    {\textrm{\Large 0}}  & \fbox{$\begin{matrix} B_2 \end{matrix}$} &  \cdots & \cdots \\
     {\textrm{\Large 0}}  &   {\textrm{\Large 0}}   &  \fbox{$\begin{matrix}  & & &\\ & B_4 &\\
   & & &\end{matrix}$}  & \cdots\\
          \textrm{\Large 0}     &   {\textrm{\Large 0}}   & {\textrm{\Large 0}} &   \ddots         & 
  \end{pmatrix}\,.
\end{equation}
Then, let us define the basis of zero- and two-switches
\begin{equation}\label{eq:closed-basis}
    \ket{k}' = \frac{1}{p} \sum_{p = 0}^{\eta/2} \sigma^p\left(\ket{i}^{\otimes k} \otimes \ket{s}^{\otimes \eta - k}\right)\,,
\end{equation}
where $k$ is even as per Lemma~\ref{lemma:1}. Note that again  the basis $\{\ket{k}'\}_{k=0,2,4,\ldots,\eta }$ is of size $\frac{\eta}{2}+1$. Now, the top left block is spanned by the vectors $\ket{0}'$ and $\ket{\eta}'$ (i.e., all $i$ and all $s$ states), whereas the $B_2$ block is spanned by $\{\ket{k}'\}_{k=2,4,\ldots,\eta-2 }$ and thus is of size $\frac{\eta-2}{2} \times \frac{\eta-2}{2}$.  Again, we show that the spectral gap corresponds to the largest eigenvalue of the $B_{2}$ matrix. In particular, we will follow a similar procedure to the open boundary case, utilizing Eq.~\eqref{eq:bounds} to find the bounds of the eigenstates of every $B_{\zeta}$ and calculating the exact largest eigenvalue of $B_2$, proving that the latter is larger than the former. 

A direct application of Eq.~\eqref{eq:layer-moment-lambda} allows us to find that for $k \neq 2, \eta-2$, the action of $\Lambda$ over these states  is (see the appendix for proof)
\footnotesize
\begin{equation}\label{eq:Lambda-2-switch}
    \Lambda \ket{k}' \!=\! \left(\frac{d^m}{d^{2m}+1}\right)^{4}\!\!\Big(\ket{k - 4}' + 4\ket{k - 2}' + 6\ket{k} + 4\ket{k + 2}' + \ket{k + 4}'\!\Big)\,,
\end{equation}
\normalsize
whereas for the cases $k = 2, \eta - 2$, we get the following equations:
\footnotesize
\begin{align}\label{eq:Lambda-2-switch-edges}
\Lambda \ket{2}' =& \left(\frac{d^m}{d^{2m}+1}\right)^{4}\Big(\frac{33}{4}|0\rangle' + 5|2 \rangle' + 4|4 \rangle' + 1|6 \rangle'\Big)\,,\nonumber\\
\Lambda \ket{\eta-2}' =& \left(\frac{d^m}{d^{2m}+1}\right)^{4}\Big(|\eta - 6 \rangle' + 4|\eta - 4 \rangle' + 5|\eta - 2 \rangle' + \frac{33}{4}|\eta \rangle'\Big).
\end{align}
\normalsize
The previous results can be used to construct $B_2$ as the following $\frac{\eta-2}{2} \times \frac{\eta-2}{2}$ matrix 
\begin{equation}\label{eq:matrix-CB-n=6}
B_2 = \left(\frac{d^m}{d^{2m}+1}\right)^{4}\begin{pmatrix}
5 & 4 & 1 & 0 & 0 & 0  & \dots & 0\\
4 & 6 & 4 & 1 & 0 & 0 & \dots & 0\\
1 & 4 & 6 & 4 & 1 & 0 & \dots & 0\\
0 & 1 & 4 & 6 & 4 & 1 & \dots & 0\\
\vdots & \vdots & \ddots & \ddots & \ddots & \ddots & \ddots & \vdots\\
0 & 0 & \dots & 1 & 4 & 6 & 4  & 1\\
0 & 0 & \dots & 0 & 1 & 4 & 6 & 4\\
0 & 0 & \dots & 0 & 0 & 1 & 4 & 5\\
\end{pmatrix}  \,.
\end{equation}

Notably, one can  verify that $B_2$ is the exact square of $B_1^\pm$ in Eq.~\eqref{eq:matrix-TT}, meaning that the eigenvalues of $B_2$ are exactly those of $B_1^\pm$, squared. From the previous, we can use Eq.~\eqref{eq:lambda_max}, to find that the spectral gap for closed boundary conditions is
\begin{align}\label{eq:lambda_max-closed}
    \lambda^{{\rm max}}&= \left(\frac{d^m}{d^{2m}+1}\right)^{4}\left(2+2\cos\left(\frac{2\pi}{\eta}\right)\right)^2\nonumber\\
    &=\left(\frac{d^m}{d^{2m}+1}\right)^4\left(2+2\cos\left(\frac{2m\pi}{n}\right)\right)^2\,.
\end{align}

Naturally, the coefficients of the leading eigenvector follows exactly from Eq.~\eqref{eq:x_max}
\begin{align}\label{eq:x_max_c}
    \ket{\lambda^{\rm max}}' =& \sin(\frac{2\pi}{\eta})\ket{2}' + \sin(\frac{4\pi}{\eta})\ket{4}' + \ldots\nonumber\\
    &+ \sin(\frac{(\eta-2)\pi}{\eta})\ket{\eta-2}'\,.
\end{align}

Finally, we note that the proof of Lemma~\ref{lemma:3} will still hold, as we are simply upper bounding the eigenvalues of the $B_\zeta$ blocks by their maximum row-sum. As such, we can compare Eqs.~\eqref{eq:bounds} and~\eqref{eq:lambda_max-closed} to find that for $\zeta\geq 4$, $\eta \geq 4$ and $m\geq1$
\begin{equation}
  \lambda_{\zeta,i}\leq \left(\frac{2d^m}{d^{2m}+1}\right)^{2\zeta}  \leq \lambda^{{\rm max}}.
\end{equation}
Note this inequality does not hold for $\zeta = 4$, $\eta = 4$ and $m = 1$, however, we again compute the largest eigenvalue of $\Lambda$ and find it to be that of $B_1$. This shows us that the largest eigenvalue of $\Lambda$ is $\lambda^{{\rm max}}$, i.e., 
\begin{equation}
   \lambda_{d,{\rm closed}}^{(n,m)}= \lambda^{{\rm max}}\,.
\end{equation}

\section{Numerical Results} 

In the previous section we have presented an exact analytical characterization for the spectral gap of random one-dimensional circuits with open and closed boundary conditions for the case where all the local gates are sampled i.i.d. from $\mathbb{U}(d^{2m})$. Here we present numerical results where we have computed the spectral gap for the unitary circuit, as well as for the following two cases:

\begin{enumerate}
    \item All local gates are sampled from the orthogonal group $G_{j,j'}^l=\mathbb{O}(d^2)$ $\forall j,j',l$, so that the ensuing unitary distribution converges in the large number of layers limit to $\mathbb{G}=\mathbb{O}(d^n)$. We refer to this as the ``\textit{orthogonal}'' case.
    \item The local gates are sampled from the orthogonal group $G_{j,j'}^l=\mathbb{O}(d^2)$ if $j,j'$ are not equal to one, and $G_{j,j'}^l=\mathbb{SP}(d^2/2)$ if  $j$ or $j'$ are equal to one. Now, the ensuing unitary distribution converges in the large number of layers limit to $\mathbb{G}=\mathbb{SP}(d^n/2)$~\cite{garcia2024architectures}. We refer to this as the ``\textit{symplectic}'' case.
\end{enumerate}

Here, we note that we tried extending our analytical results to these two cases, but the block diagonal structure becomes more complex, and it is unclear to us how to compute the eigenvalues associated to each block. 

For simplicity, we consider the cases when $m=1$, and $d=2$. That is, a random circuit where the local gates act on neighboring pairs of qubits. Our simulations are based on constructing the action of the second moment operator with a procedure similar to that in Eq.~\eqref{eq:layer-moment-lambda}. In particular, for Case 1 we observe that the single-layer moment operator $\TC_{\EC_1}^{(2)}$ can be represented by the $3^\eta\times 3^\eta$ operator 

\begin{equation}
 \Lambda_{\mathbb{O}}=\left(\prod_{j = 1}^{\frac{\eta}{2}} \tilde{A}_{2j-1,2j}\right)(\tilde{A}_{\eta,1})^\Delta  \left(\prod_{j = 1}^{\frac{\eta-2}{2}} \tilde{A}_{2j,2j+1} \right)\,, \nonumber
\end{equation}

where  $\tilde{A}$ is  the $9 \times 9$ matrix 
\begin{equation}\label{eq:matrix-P-2} 
\tilde{A} = 
       \begin{pmatrix}
1 & 7/18 & 7/18 & 7/18 & 0 & 1/18 & 7/18 & 1/18 & 0\\
0 & 0 & 0 & 0 & 0 & 0 & 0 & 0 & 0\\
0 & 0 & 0 & 0 & 0 & 0 & 0 & 0 & 0\\
0 & 0 & 0 & 0 & 0 & 0 & 0 & 0 & 0\\
0 & 7/18 & 1/18 & 7/18 & 1 & 7/18 & 1/18 & 7/18 & 0\\
0 & 0 & 0 & 0 & 0 & 0 & 0 & 0 & 0\\
0 & 0 & 0 & 0 & 0 & 0 & 0 & 0 & 0\\
0 & 0 & 0 & 0 & 0 & 0 & 0 & 0 & 0\\
0 & 1/18 & 7/18 & 1/18 & 0 & 7/18 & 7/18 & 7/18 & 1\\
\end{pmatrix} \,.
\end{equation}
We refer the reader to~\cite{braccia2024computing} for additional details on how this matrix is obtained. Then, for Case 1 the single-layer moment operator $\TC_{\EC_1}^{(2)}$ can be represented by the $3\cdot2^{\eta-1}\times 3\cdot2^{\eta-1}$ operator 

\begin{equation}
 \Lambda_{\mathbb{SP}}= \tilde{A}_{1,2} \left(\prod_{j = 2}^{\frac{\eta}{2}} A_{2j-1,2j}\right)(\tilde{A}_{\eta,1})^\Delta  \left(\prod_{j = 1}^{\frac{\eta-2}{2}} A_{2j,2j+1} \right)\,, \nonumber
\end{equation}
where $A$ and $\widetilde{A}$ are defined in Eqs.~\eqref{eq:matrix-P} and~\eqref{eq:matrix-P-2}, respectively (see also~\cite{garcia2024architectures}).

\begin{figure}[t]
    \centering
    \includegraphics[width=.95\columnwidth]{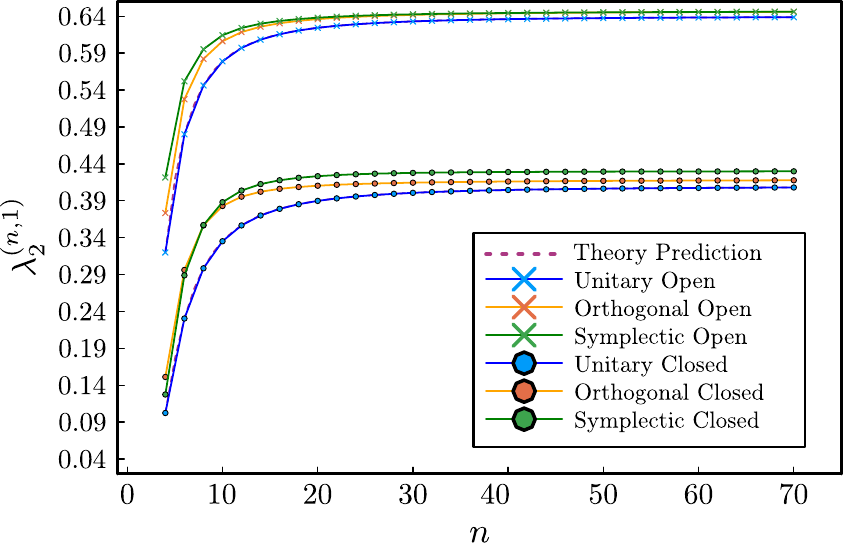}
    \caption{\textbf{Spectral gap for unitary, orthogonal and symplectic circuits.} We plot the spectral gap versus the number of qubits for different groups, as well as different boundary conditions on a single layered random quantum circuit. }
    \label{fig:numerics}
\end{figure}

Using Density Matrix Renormalization Group (DMRG)-based techniques (described in detail in the Appendix~\ref{appendixDMRG}), we numerically compute the spectral gap for the unitary, orthogonal and symplectic cases, and for open and closed boundary conditions. The results are shown in Fig.~\ref{fig:numerics}. Here we can first see that the numerical predictions exactly match our theoretical results for the unitary case. Moreover, we can see that regardless of the boundary conditions, the unitary circuit has the smallest spectral gap, indicating that it will require the least number of layers to be an approximate design over its associated group. Interestingly, we see that in the large $n$ limit, the symplectic case has the largest spectral gap. For open boundary conditions it appears that the spectral gaps for orthogonal and symplectic circuits match (for large $n$), but for closed boundary conditions we see that  the symplectic gap is larger than that of the orthogonal. Intriguingly, this behavior appears to not be monotonic as for small number of qubits, the symplectic spectral gap is actually smaller than the orthogonal one. Then, we can numerically check that in the orthogonal and symplectic cases, the spectral gaps are not related by a square root as in the unitary case. Indeed, we can take this fact as an indicator of the increased complexity of the symplectic and orthogonal cases.

\section{Discussion}

In this work we provide an exact characterization of the spectral gap for the second moment of random unitary quantum circuits on a one-dimensional lattice with nearest neighboring gates. We use Weingarten calculus to determine the upper triangular block-matrix form of the moment operator when acting on the bases arising from the local gate's commutant. An immediate consequence of our work is that bounds based on the one-dimensional spectral gap (e.g., for other architectures~\cite{belkin2023approximate}) can be immediately tightened, as the asymptotic scalings obtained from our exact formulae are better than those based on previous works. Moreover, we envision that knowledge of the upper triangular form, as well as the eigenvector associated to the spectral gap, could be useful in other contexts beyond those considered here (e.g. for error mitigation~\cite{hu2024demonstration}).

Then, we note that our work raises the following question: \textit{Can the techniques be extrapolated to consider larger moments, or other groups in the local gates?} Given that one ultimately needs to work with a basis determined by the local commutants, the calculations will naturally become more difficult for larger moments and for local gates sampled from subgroups of $\mathbb{U}(d^m)$, as any of those changes increase the local commutant's dimension. For instance, when the local gates are orthogonal, the commutant dimension changes from two to three, and already this small, albeit non-trivial, change makes the characterization of the moment operator much more difficult. Despite this fact, we hope that our work will serve as a starting point for other studies and that the tools presented here can serve as blueprints for other works. 

Next, we would like to highlight the recent work of~\cite{schuster2024random}, which shows that one-dimensional quantum circuits composed of gates acting on sets of $m$ qudits can converge to designs in a logarithmic number of layers. Such improvement is achieved by studying the relative error, rather than the additive one. Here, it is not hard to envision that the exact characterization of the eigenspaces of the single layer moment operator can be combined with the techniques in~\cite{schuster2024random} to further improve the convergence scaling of the second-moment operator. We leave such exploration for future work. 

To finish, here we briefly relate our spectral-gap view to recent information-theoretic approaches to symmetry testing. Our spectral–gap perspective quantifies, at the level of the $t=2$ moment operator, how rapidly random local circuits approach the corresponding Haar measure, and it makes the group dependence (unitary vs.\ orthogonal vs.\ symplectic) explicit through exact gaps in the cases we analyze. Recent information–theoretic works take a complementary approach by casting the task as hypothesis testing between Haar–random ensembles (e.g., $\mathbb{U}(d)$ versus $\mathbb{O}(d)$) and characterizing optimal discrimination protocols and error exponents via representation–theoretic performance operators and max–relative–entropy–type quantities. For instance, as shown in~\cite{chen2025hypothesis,hayashi2025predicting} parallel (non–adaptive) strategies can attain optimality.  Our exact second–moment gaps thus provide operational, group–dependent contraction factors that could inform—or be converted into—bounds on distinguishing advantage or query complexity when only low–moment information is available. Making this relation precise (for example, by connecting our $t=2$ gap to the exponents governing type–II error in optimal tests, or by extending the comparison to higher moments) is an interesting open problem for future work.

\section{Acknowledgments}

We thank Daniel Belkin, Shivan Mittal, Aroosa Ijaz and Diego Garc\'ia-Mart\'in for useful and insightful conversations. A.E.D., Pablo B., Paolo B.  and L.C. were supported by the Laboratory Directed Research and Development (LDRD) program of Los Alamos National Laboratory (LANL) under project numbers 20230527ECR and 20230049DR. Pablo B. acknowledges constant support from DIPC. M.C. acknowledges support from LANL's ASC Beyond Moore’s Law project. This work was also supported by the U.S. Department of Energy, Office of Science, Office of Advanced Scientific Computing Research, under Computational Partnerships program.

\bibliography{quantum}

\clearpage
\newpage
\onecolumngrid
\appendix
\section*{Appendix}
\setcounter{lemma}{0}

In this appendix we present the proof for Lemmas~\ref{lemma:1}--\ref{lemma:3}, as well as for Eqs.~\eqref{eq:Lambda-1-switch}, ~\eqref{eq:Lambda-2-switch} and~\eqref{eq:Lambda-2-switch-edges}. We also discuss details of our numerical simulations.

\section{Proof of Lemma~\ref{lemma:1}}
\begin{lemma}\label{lemma:1-ap}
    An eigenvector of $\Lambda$ must also be an eigenvector of $A^{\otimes \frac{\eta}{2}}$ and thus can only be written in terms of basis vectors $\ket{\alpha_0}^{\otimes k_0} \otimes \ket{\alpha_1}^{\otimes k_1} \otimes \dots \otimes \ket{\alpha_\zeta}^{\otimes k_\zeta}$ such that $k_j$ is even $\forall j$.
\end{lemma}

\begin{proof} Given that each $A$ is a projector, then so is $A^{\otimes \frac{\eta}{2}}$. With that we have
\begin{align}
 A^{\otimes \frac{\eta}{2}} \Lambda \ket{\psi} = A^{\otimes \frac{\eta}{2}} A^{\otimes \frac{\eta}{2}} (\id_2 \otimes A^{\frac{\eta - 2}{2}} \otimes \id_2) \ket{\psi} = A^{\otimes \frac{\eta}{2}} (\id_2 \otimes A^{\frac{\eta - 2}{2}} \otimes \id_2) \ket{\psi} = \Lambda \ket{\psi} \,.
\end{align}
Thus, $\Lambda\ket{\psi}$ is an eigenvector of $A^{\otimes \frac{\eta}{2}}$ for any $\ket{\psi}$, so we can say that $\Lambda$ projects into the eigenspace of $A^{\otimes \frac{\eta}{2}}$. Moreover, since $A^{\otimes \frac{\eta}{2}}$ is a projector its eigenvalues can only be equal to $1$ or $0$. As we are only concerned with non-zero eigenvalues of $\Lambda$ then we can assume that the eigenvector $\ket{\psi}$ of $\Lambda$ is also an eigenvector of $A^{\otimes \frac{\eta}{2}}$ with eigenvalue $1$. We can see that an eigenvector of $A^{\otimes \frac{\eta}{2}}$ with eigenvalue $1$ must be a linear combination of basis terms of the form $|\alpha_1 \alpha_2 \alpha_3 \dots \alpha_\eta \rangle$ such that for $j$ odd, $\alpha_j = \alpha_{j+1}$. This is because $A^{\otimes \frac{\eta}{2}}$ applies an $A$ matrix to neighboring pairs of qubits, $\alpha_j \alpha_{j+1}$, such that $j$ is odd. Therefore an eigenvector of $A^{\otimes \frac{\eta}{2}}$ (and thus of $\Lambda$) with a non-zero eigenvalue, must be in the form $\ket{\alpha_0}^{\otimes k_0} \otimes \ket{\alpha_1}^{\otimes k_1} \otimes \dots \otimes \ket{\alpha_\zeta}^{\otimes k_\zeta}$ where $k_j$ is even $\forall j$.
\end{proof}

\section{Proof of Lemma~\ref{lemma:2}} 
\begin{lemma}\label{lemma:2-ap}
If a vector is a linear combination of basis terms with $\zeta$, or less, switches then the action on this vector by $\Lambda$ will result in a linear combination of basis states with $\zeta$, or less, switches. 
\end{lemma}

\begin{proof}
 Let $\ket{\psi}$ be a computational basis state with $\zeta$ switches, i.e., $\ket{\psi} = \ket{\alpha_0}^{\otimes k_0} \otimes \ket{\alpha_1}^{\otimes k_1} \otimes \dots \otimes \ket{\alpha_{\zeta}}^{\otimes k_{\zeta}}$. Due to Lemma~\ref{lemma:1}, $k_j$ is even for all $j$. To study the action of $\Lambda$ on this state, we first apply the $(\id_2 \otimes A^{\frac{\eta - 2}{2}} \otimes \id_2)$ matrix. It is clear that this will leave most of the state unchanged, except at the switches. Let us analyze the subsystem of the first switch in the state, i.e., $\ket{\alpha_0}^{\otimes k_0} \otimes \ket{\alpha_1}^{\otimes k_1}$. Without loss of generality we let $\alpha_0 = i$ then $\alpha_1 = s$. To further simplify the analysis we examine the system of 4 qubits that define the switch: $\ket{iiss}$. Since $k_0$ and $k_1$ are even, applying $(\id_2 \otimes A^{\frac{\eta - 2}{2}} \otimes \id_2)$ applies an $A$ gate on the middle 2 qubits so we get
\begin{equation}
    (\id_2 \otimes A \otimes \id_2)\ket{iiss} = \frac{2}{5}(\ket{iiis} + \ket{isss})\,.
\end{equation}
Next, we multiply by the $A^{\otimes \frac{\eta}{2}}$ matrix which will apply an $A$ gate to the first and last pairs of qubits in the considered reduced subsystem
\begin{equation}
    A^{\otimes 2} \frac{2}{5}(\ket{iiis} + \ket{isss}) = \frac{4}{25}(\ket{iiii} + \ket{iiss} + \ket{iiss} + \ket{ssss})\,.
\end{equation}
Hence, we see that $\Lambda$ acts on the subsystem of the first switch as $(\ket{\alpha_0}^{\otimes k_0} \otimes \ket{\alpha_1}^{\otimes k_1}) \rightarrow \frac{4}{25}[\ket{\alpha_0}^{\otimes k_0 - 2} \otimes \ket{\alpha_1}^{\otimes k_1 + 2} + 2\ket{\alpha_0}^{\otimes k_0} \otimes \ket{\alpha_1}^{\otimes k_1} + \ket{\alpha_0}^{\otimes k_0 + 2} \otimes \ket{\alpha_1}^{\otimes k_1 - 2}]$.  
Thus, depending on $k_0, k_1$, the first switch will be taken to a superposition of terms with either a single switch or no switches. This same procedure can be applied to all of the switches in the state $\ket{\psi}$ and thus the application of $\Lambda$ on a state with $\zeta$ switches will result in a linear combination of basis states with at most $\zeta$ switches.
\end{proof}

\section{Proof of Eq.~\eqref{eq:Lambda-1-switch}: Action of $\Lambda$ on 1-switch terms} \label{appendix:lambda-proof}

We have from Eq.~\eqref{eq:lambda-simple} that $\Lambda = A^{\otimes \frac{\eta}{2}}(\id_2 \otimes A^{\frac{\eta - 2}{2}} \otimes \id_2)$. Let us express an arbitrary 1-switch vector for the open boundary conditions as $\ket{k}_+ = \ket{i}^{\otimes k - 2} \otimes \ket{iiss} \otimes \ket{s}^{\otimes \eta - k - 2}$, where $k \neq 0, \eta$. Multiplying this vector by  $\id_2 \otimes A^{\frac{\eta - 2}{2}} \otimes \id_2$ leads to
\begin{align}
    &\frac{d^m}{d^{2m} + 1}\left(\ket{i}^{\otimes k - 2} \otimes \ket{iiis} \otimes \ket{s}^{\otimes \eta - k - 2}\right) + \frac{d^m}{d^{2m} + 1}\left(\ket{i}^{\otimes k - 2} \otimes \ket{isss} \otimes \ket{s}^{\otimes \eta - k - 2}\right) \nonumber\\
    &= \frac{d^m}{d^{2m} + 1} \left(\ket{i}^{\otimes k - 2} \otimes (\ket{iiis} + \ket{isss}) \otimes \ket{s}^{\otimes \eta - k - 2}\right) \label{eq:app-c-action},
\end{align}
where we have used the fact that $k$ is even as per Lemma~\ref{lemma:1}. Then, multiplying this by $A^{\otimes \frac{\eta}{2}}$
\footnotesize
\begin{align}
    &\left(\frac{d^m}{d^{2m} + 1}\right)^2\left(\ket{i}^{\otimes k - 2} \otimes \ket{iiii} \otimes \ket{s}^{\otimes \eta - k - 2}\right) + 2\left(\frac{d^m}{d^{2m} + 1}\right)^2\left(\ket{i}^{\otimes k - 2} \otimes \ket{iiss} \otimes \ket{s}^{\otimes \eta - k - 2}\right)+ \left(\frac{d^m}{d^{2m} + 1}\right)^2\left(\ket{i}^{\otimes k - 2} \otimes \ket{ssss} \otimes \ket{s}^{\otimes \eta - k - 2}\right) \nonumber\\
    &= \left(\frac{d^m}{d^{2m} + 1}\right)^2 \ket{i}^{\otimes k - 2} \otimes \left(\ket{iiii} + 2\ket{iiss} + \ket{ssss}\right) \otimes \ket{s}^{\otimes \eta - k - 2}= \left(\frac{d^m}{d^{2m} + 1}\right)^2\left(\ket{k +2}_+ + 2\ket{k}_+ + \ket{k - 2}_+\right).\nonumber
\end{align}
\normalsize
The result for a general $\ket{k}_-$ follows accordingly. Thus we find that the action of $\Lambda$ on a general 1-switch term $\ket{k}_\pm$, i.e., when $2 \leq k \leq \eta - 2$, is
\begin{equation}
    \Lambda\ket{k}_{\pm} = \left(\frac{d^m}{d^{2m}+1}\right)^{2}\left(\ket{k +2}_{\pm} + 2\ket{k}_{\pm} + \ket{k - 2}_{\pm}\right) \nonumber
\end{equation}

For $k = 0, \eta$ we have the all $s$ and all $i$ vectors, respectively, which are eigenstates of $A$ and therefore so $\Lambda \ket{0}_\pm = \ket{0}_\pm$ and $\Lambda \ket{\eta}_\pm = \ket{\eta}_\pm$. $\qed$

\section{Proof of Eqs.~\eqref{eq:Lambda-2-switch} and~\eqref{eq:Lambda-2-switch-edges}: Action of $\Lambda$ on 2-switch terms}

Similar to the open boundary case we want to study the action of $\Lambda$ on a generic 2-switch basis vector. Thus, let us define the vector $\ket{\psi} = \ket{i} \otimes \ket{i}^{\otimes k - 3} \otimes \ket{iiss} \otimes \ket{s}^{\otimes \eta-k-3} \otimes \ket{s}$ with $4 \leq k \leq \eta - 4$.  The action of $(\id_2 \otimes A^{\frac{\eta - 2}{2}} \otimes \id_2)$  follows from Eq.~\eqref{eq:app-c-action}, so that
\begin{align}
    \id_2 \otimes A^{\frac{\eta - 2}{2}} \otimes \id_2\ket{\psi} = \frac{d^m}{d^{2m} + 1}\left (\ket{i} \otimes\ket{i}^{\otimes k - 3} \otimes (\ket{iiis} + \ket{isss}) \otimes \ket{s}^{\otimes \eta - k - 3} \otimes\ket{s}\right) \nonumber\,.
\end{align}

Now we multiply by an A matrix acting non-trivially on the first and last qubit, resulting in a state of the form
\begin{align}
    &\left(\frac{d^m}{d^{2m} + 1}\right)^2 \Big(\ket{ii} \otimes\ket{i}^{\otimes k - 4} \otimes \ket{iiis} \otimes \ket{s}^{\otimes \eta - k - 4} \otimes\ket{si} + \ket{si} \otimes\ket{i}^{\otimes k - 4} \otimes \ket{iiis} \otimes \ket{s}^{\otimes \eta - k - 4} \otimes\ket{ss} \nonumber\\
    &+ \ket{ii} \otimes\ket{i}^{\otimes k - 4} \otimes \ket{isss} \otimes \ket{s}^{\otimes \eta - k - 4} \otimes\ket{si}+ \ket{si} \otimes\ket{i}^{\otimes k - 4} \otimes \ket{isss} \otimes \ket{s}^{\otimes \eta - k - 4} \otimes\ket{ss}\Big) \nonumber\,.
\end{align}
As we can see each term has 2 switches. The action of $A^{\otimes \frac{\eta}{2}}$ on this vector of four terms results in a vector of sixteen terms as shown below
\begin{align}\label{eq:closed-general}
    \left(\frac{d^m}{d^{2m} + 1}\right)^4 \Big(&\ket{ii} \otimes\ket{i}^{\otimes k - 4} \otimes \ket{iiii} \otimes \ket{s}^{\otimes \eta - k - 4} \otimes\ket{ii} + \ket{ii} \otimes\ket{i}^{\otimes k - 4} \otimes \ket{iiii} \otimes \ket{s}^{\otimes \eta - k - 4} \otimes\ket{ss}\nonumber\\
    &+ \ket{ii} \otimes\ket{i}^{\otimes k - 4} \otimes \ket{iiss} \otimes \ket{s}^{\otimes \eta - k - 4} \otimes\ket{ii} + \ket{ii} \otimes\ket{i}^{\otimes k - 4} \otimes \ket{iiss} \otimes \ket{s}^{\otimes \eta - k - 4} \otimes\ket{ss}\nonumber\\
    &+ \ket{ii} \otimes\ket{i}^{\otimes k - 4} \otimes \ket{iiii} \otimes \ket{s}^{\otimes \eta - k - 4} \otimes\ket{ss} + \ket{ss} \otimes\ket{i}^{\otimes k - 4} \otimes \ket{iiii} \otimes \ket{s}^{\otimes \eta - k - 4} \otimes\ket{ss}\nonumber\\
    &+ \ket{ii} \otimes\ket{i}^{\otimes k - 4} \otimes \ket{iiss} \otimes \ket{s}^{\otimes \eta - k - 4} \otimes\ket{ss} + \ket{ss} \otimes\ket{i}^{\otimes k - 4} \otimes \ket{iiss} \otimes \ket{s}^{\otimes \eta - k - 4} \otimes\ket{ss}\nonumber\\
    &+ \ket{ii} \otimes\ket{i}^{\otimes k - 4} \otimes \ket{ssss} \otimes \ket{s}^{\otimes \eta - k - 4} \otimes\ket{ii} + \ket{ii} \otimes\ket{i}^{\otimes k - 4} \otimes \ket{ssss} \otimes \ket{s}^{\otimes \eta - k - 4} \otimes\ket{ss}\nonumber\\
    &+ \ket{ii} \otimes\ket{i}^{\otimes k - 4} \otimes \ket{iiss} \otimes \ket{s}^{\otimes \eta - k - 4} \otimes\ket{ii} + \ket{ii} \otimes\ket{i}^{\otimes k - 4} \otimes \ket{iiss} \otimes \ket{s}^{\otimes \eta - k - 4} \otimes\ket{ss}\nonumber\\
    &+ \ket{ii} \otimes\ket{i}^{\otimes k - 4} \otimes \ket{ssss} \otimes \ket{s}^{\otimes \eta - k - 4} \otimes\ket{ss} + \ket{ss} \otimes\ket{i}^{\otimes k - 4} \otimes \ket{ssss} \otimes \ket{s}^{\otimes \eta - k - 4} \otimes\ket{ss}\nonumber\\
    &+ \ket{ii} \otimes\ket{i}^{\otimes k - 4} \otimes \ket{iiss} \otimes \ket{s}^{\otimes \eta - k - 4} \otimes\ket{ss} + \ket{ss} \otimes\ket{i}^{\otimes k - 4} \otimes \ket{iiss} \otimes \ket{s}^{\otimes \eta - k - 4} \otimes\ket{ss}\Big ) \nonumber\,.
\end{align}
The previous equation can be simplified by consolidating blocks of $i$ and $s$
\begin{align}
    \left(\frac{d^m}{d^{2m} + 1}\right)^4 \Big(&\ket{i}^{\otimes k + 2} \otimes \ket{s}^{\otimes \eta - k - 4} \otimes\ket{ii}+ 2 \ket{i}^{\otimes k + 2}\otimes \ket{s}^{\otimes \eta - k - 2}+ 2\ket{i}^{\otimes k} \otimes \ket{s}^{\otimes \eta - k - 2} \otimes\ket{ii}+ 4\ket{i}^{\otimes k}\otimes \ket{s}^{\otimes \eta - k}\nonumber\\
    &+ \ket{ss} \otimes\ket{i}^{\otimes k} \ket{s}^{\otimes \eta - k - 2}+ \ket{i}^{\otimes k - 2} \otimes \ket{s}^{\otimes \eta - k} \otimes\ket{ii}+ 2\ket{i}^{\otimes k - 2} \otimes \ket{s}^{\otimes \eta - k + 2}+ 2\ket{ss} \otimes\ket{i}^{\otimes k - 2} \otimes \ket{s}^{\otimes \eta - k}\nonumber\\
    &+ \ket{ss} \otimes\ket{i}^{\otimes k - 4} \otimes \ket{s}^{\otimes \eta - k + 2}\Big)\nonumber\,.
\end{align}

As mentioned in the main text, the action of $\Lambda$ commutes with that of $\sigma^p$ for all $p$, where $\sigma$ is a two-qubit site translation. Thus, we can further simplify the expression by writing it in terms of our basis $\{ \ket{k}' \}$. This lets us combine terms that are related by cyclical permutation, i.e., terms that have the same number of $i$ (or $s$) states. As such, we find that 
\begin{equation}
    \Lambda \ket{k}' = \left(\frac{d^m}{d^{2m} + 1}\right)^4 \left(\ket{k - 4}' + 4\ket{k - 2}' + 6\ket{k}' + 4\ket{k + 2}' + \ket{k + 4}'\right) \nonumber\,.
\end{equation}

To finish, we consider  the cases where $k = 2, \eta - 2$. For $k = 2$ we have a vector $\ket{ii} \otimes \ket{s}^{\otimes \eta - 2}$. Multiplying by $(\id_2 \otimes A^{\frac{\eta - 2}{2}} \otimes \id_2)$ will follow as in Appendix~\ref{appendix:lambda-proof} resulting in
\begin{equation}
    \frac{d}{d^{2m} + 1}\left(\ket{iii} \otimes \ket{s}^{\otimes \eta - 3} + \ket{i} \otimes \ket{s}^{\otimes \eta - 1}\right) \nonumber\,.
\end{equation}
Then applying an $A$ gate to the first and last qubit leads to
\begin{equation}
    \left(\frac{d}{d^{2m} + 1}\right)^2\left(\ket{iii} \otimes \ket{s}^{\otimes \eta - 4} \otimes \ket{i} + \ket{i} \otimes \ket{s}^{\otimes \eta - 2} \otimes \ket{i} + \ket{sii} \otimes \ket{s}^{\otimes \eta - 3} + \ket{s}^{\otimes \eta}\right) \nonumber\,.
\end{equation}

Next we apply the operation $A^{\otimes \frac{\eta}{2}}$ which results in the following state
\begin{align}
    \left(\frac{d}{d^{2m} + 1}\right)^4\Big(&\ket{iiii} \otimes \ket{s}^{\otimes \eta - 6} \otimes \ket{ii} + \ket{ii} \otimes \ket{s}^{\otimes \eta - 4} \otimes \ket{ii} + \ket{iiii} \otimes \ket{s}^{\otimes \eta - 4} + \ket{ii} \otimes \ket{s}^{\otimes \eta - 2} \nonumber\\
    &+ \ket{ii} \otimes \ket{s}^{\otimes \eta - 4} \otimes \ket{ii} + \ket{ii} \otimes \ket{s}^{\otimes \eta - 2} + \ket{s}^{\otimes \eta - 2} \otimes \ket{ii} + \ket{s}^{\otimes \eta}\nonumber\\
    &+ \ket{iiii} \otimes \ket{s}^{\otimes \eta - 4} + \ket{ssii} \otimes \ket{s}^{\otimes \eta - 4} + \ket{ii} \otimes \ket{s}^{\otimes \eta - 2} + \ket{s}^{\otimes \eta}\Big)+ \left(\frac{d}{d^{2m} + 1}\right)^2\ket{s}^{\otimes \eta} \nonumber\,.
\end{align}
We can use this result to expressing the action on the basis $\{ \ket{k}' \}$ as
\begin{equation}
    \Lambda \ket{2}' = \left(\frac{d^m}{d^{2m}+1}\right)^{4}\Big(\frac{33}{4}|0\rangle' + 5|2 \rangle' + 4|4 \rangle' + 1|6 \rangle'\Big)\,.
\end{equation}

Using  the symmetry of $i$ and $s$ basis states, the case with $k = \eta - 2$ follows immediately as
\begin{align}
\Lambda \ket{\eta-2}' =& \left(\frac{d^m}{d^{2m}+1}\right)^{4}\Big(|\eta - 6 \rangle' + 4|\eta - 4 \rangle' + 5|\eta - 2 \rangle' + \frac{33}{4}|\eta \rangle'\Big)\,.\nonumber \qed
\end{align}

\section{Proof Lemma~\ref{lemma:3}}

\begin{lemma}\label{lemma:3-append}
    Let $\{\lambda_{\zeta,i} \}$ be the  eigenvalues of $B_\zeta$. Then, the following  upper  bound holds for all $i$, and $\zeta \geq 1$
\begin{equation}\label{eq:bounds-append}
    \lambda_{\zeta,i}\leq \left(\frac{2d^m}{d^{2m}+1}\right)^{2\zeta} \,.
\end{equation}
\end{lemma}

\begin{proof}
We can find upper bounds on the eigenvalues of the $B_\zeta$ submatrices by finding the largest row/column sum. To do this we look at how $B_\zeta$ sends $\zeta$-switch terms to other $\zeta$-switch terms. Let $\ket{x^{(\zeta)}}$ be a vector with $\zeta$ switches. Then $(\id_2 \otimes A^{\frac{\eta - 2}{2}} \otimes \id_2)\ket{x^{(\zeta)}} = (\frac{d^m}{d^{2m} + 1})^\zeta \sum_{j = 1}^{2^\zeta} \ket{y^{(\zeta')}_j}$. Now $A^{\otimes \frac{\eta}{2}} (\id_2 \otimes A^{\frac{\eta - 2}{2}} \otimes \id_2)\ket{x^{(k)}} = (\frac{d^m}{d^{2m} + 1})^\zeta \sum_{i = 1}^{2^\zeta} (\frac{d^m}{d^{2m} + 1})^{\zeta'} \sum_{j = 1}^{2^{\zeta'}} \ket{y^{(\zeta'')}_{ij}}$ for $\zeta'' \leq \zeta' \leq \zeta$ by Lemma~\ref{lemma:2}. Then the maximum rowsum of $B_\zeta$ will be given when $\zeta'' = \zeta' = \zeta$ and is given by 
\begin{equation}\label{eq:bounds-old-append}
    \left(\frac{d^m}{d^{2m} + 1}\right)^{2{\zeta}} \cdot 2^{2\zeta} = \left(\frac{2d^m}{d^{2m} + 1}\right)^{2\zeta}\,.
\end{equation}

\end{proof}

\section{Details on the DMRG numerical simulations}
\label{appendixDMRG}
Here we explain the techniques used to numerically compute the spectral gaps via  Density Matrix Renormalization Group (DMRG) for up to $n = 70$ qubits.  First of all, we realize that the layered structure of the $\Lambda$ matrix can be exploited to simplify the calculation. In particular, we express $\Lambda=L_1\cdot L_2$, where $L_1$ contains the matrices acting on the local Hilbert spaces indexed by $(2,3)$, and $L_2$ denotes the one containing the matrices acting on sites indexed by $(1,2)$.

As previously mentioned, $L_1$ and $L_2$ are projectors and, hence, $L_1L_2$ and $L_1L_2L_1$ have the same spectrum. The operator $L_1L_2L_1$ can now be converted into a Hermitian matrix by means of a non-unitary change of basis $g$. In particular, we apply such a local change of basis to each $A$ matrix as  $(g^{-1} \otimes g^{-1})A(g \otimes g)$. Thus, we are left with a Hermitian operator that can be efficiently turned into a Matrix Product Operator (MPO). The efficiency of the method lies in a low and constant bond dimension for that MPO. 

Given that  we are interested in retrieving the largest non-one eigenvalue of $\Lambda$, we need to remove from this operator the projectors onto the second-order commutant of the group that the circuit belongs to. For instance, in the unitary case we need to subtract the projectors onto the   all $i$ and all $s$ states. This can be achieved by redefining the operator of interest as $H =  L_1L_2L_1 - \sum_j v_j^\dagger v_j$, where $v_j$ stand for the vectorized elements of the second-order commutant of the group to which the circuit's distribution converges. The operator $H$ is the one employed for the DMRG procedure, which finds the ground states of $-H$. During our simulations, we detect that a bond dimension of $\chi = 70-80$ for the Matrix Product State (MPS) describing the approximation to the ground state is enough to reach convergence up to 8 digits in energy in all the cases we present in this work, up to system sizes of $70$ qutrits in the case of the orthogonal and symplectic groups, and over $100$ qubits in the case of the unitary group.

\end{document}